\documentclass[a4paper]{article}
\usepackage{amssymb}
\usepackage{amsfonts}
\usepackage{amsmath}
\usepackage{geometry}
\usepackage{graphicx}

\newtheorem{theorem}{Theorem}
\newtheorem{acknowledgement}[theorem]{Acknowledgement}

\newtheorem{proposition}[theorem]{Proposition}
\newtheorem{remark}[theorem]{Remark}

\newenvironment{proof}[1][Proof]{\noindent\textbf{#1.} }{\ \rule{0.5em}{0.5em}}
\input{tcilatex}
\begin{document}

\title{Forced oscillations of a body attached to a viscoelastic rod of
fractional derivative type}
\author{Teodor M. Atanackovic%
\begin{footnote}
{Department of Mechanics, Faculty of Technical Sciences,
University of Novi Sad, Trg D. Obradovica, 6, 21000 Novi Sad,
Serbia, atanackovic@uns.ac.rs}
\end{footnote}, Stevan Pilipovic%
\begin{footnote}
{Department of Mathematics, Faculty of Natural Sciences and
Mathematics, University of Novi Sad, Trg D. Obradovica, 4, 21000
Novi Sad, Serbia, stevan.pilipovic@dmi.uns.ac.rs}
\end{footnote} and Dusan Zorica%
\begin{footnote}
{Mathematical Institute, Serbian Academy of Sciences and Arts,
Kneza Mihaila 36, 11000 Beograd, Serbia, dusan\textunderscore
zorica@mi.sanu.ac.rs}
\end{footnote}}
\maketitle

\begin{abstract}
\noindent We study forced oscillations of a rod with a body attached to its
free end so that the motion of a system is described by two sets of
equations, one of integer and the other of the fractional order. To the
constitutive equation we associate a single function of complex variable
that plays a key role in finding the solution of the system and in
determining its properties. This function could be defined for a linear
viscoelastic bodies of integer/fractional derivative type.

\bigskip

\end{abstract}

\section{Introduction}

In this paper we continue our recent work on the dynamics of viscoelastic
rods described through the fractional derivative type equations, presented
in \cite{a,AKOZ,APZ-3,APZ-4,APZ-5}. A problem that we shall be dealing with
in the present paper is forced oscillations of a body attached to a
viscoelastic rod with comparable masses. A rod-body system is shown in
Figure \ref{fig-00}.
\begin{figure}[h]
\centering
\includegraphics[scale=1]{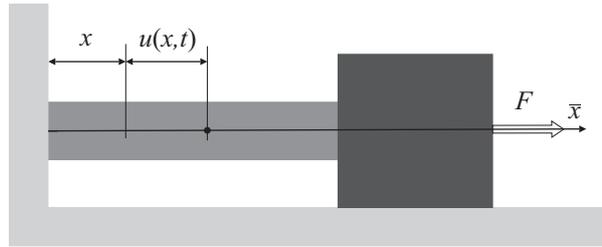}
\caption{System rod-body.}
\label{fig-00}
\end{figure}

Our new approach in the investigation of the dynamics of linear viscoelastic
rods of fractional type, proposed in \cite{APZ-5}, is based on the
properties of a specially defined function of complex variable $M$ (see (\ref%
{M})). Function $M$ is associated with the Laplace transform of the
constitutive equation for the material of a rod. It is defined similarly as
complex moduli (a quantity obtained after application of Fourier transform
to constitutive equation). By considering two cases of constitutive
equations, we show in this paper (that is a continuation of \cite{APZ-5})
that the proposed approach can be adventitiously used in dynamics of
viscoelastic rods.

We find an explicit form of the solution as a convolution of a forcing
function and a solution kernel. Moreover, we present numerical examples,
corresponding to two common cases of constitutive equations. We analyze
forced oscillations of a system consisting of a viscoelastic rod of
fractional-order type and a body attached to its end. Thus, the cases of the
dynamics of an elastic rod and of a light rod (mass of a rod is negligible)
are special cases of our analysis. We refer to \cite{novacki} for the
analysis of oscillations of an elastic rod with the mass attached to its
end.

In \cite{a}, we analyzed forced oscillations of a body attached to a
viscoelastic rod, described by a fractional distributed-order model. We
assumed that the mass of a rod is negligible compared to the mass of a body.
Similarly as in \cite{a}, we analyzed in \cite{APZ-3,APZ-4} the wave
propagation in a viscoelastic solid-like rod of finite length with one of
its ends fixed to a rigid wall. We considered two cases (i.e. two types of
boundary conditions): the case when there is a prescribed displacement and
the case when there is prescribed stress on rod's free end. Similar problem
of wave propagation was analyzed in \cite{AKOZ} for a rod made of
viscoelastic fluid-like material. The present paper is closely related to
\cite{APZ-5} where we analyze a more general form of a constitutive
equation. Actually, in \cite{APZ-5} we gave a theoretical background which
we use here and analyze two models which will be described below.

Let $m$ be the mass of a body attached to a rod. The length of the rod in
undeformed state is $L$ and its axis, at the initial time moment as well as
during the motion, coincides with the $\bar{x}$ axis, see Figure \ref{fig-00}%
. Let $x$ denote a position of a material point of the rod at the initial
time $t_{0}=0.$ The position of this point at the time $t>0$ is $x+u\left(
x,t\right) .$ The equations of motion of the rod-body system are
\begin{gather}
\frac{\partial }{\partial x}\sigma \left( x,t\right) =\rho \frac{\partial
^{2}}{\partial t^{2}}u\left( x,t\right) ,\;\;\;\;\varepsilon \left(
x,t\right) =\frac{\partial }{\partial x}u\left( x,t\right) ,\;\;x\in \left[
0,L\right] ,\;t>0,  \label{EM-SM-dim} \\
\int_{0}^{1}\phi _{\sigma }\left( \gamma \right) \,{}_{0}\mathrm{D}%
_{t}^{\gamma }\sigma \left( x,t\right) \mathrm{d}\gamma =E\int_{0}^{1}\phi
_{\varepsilon }\left( \gamma \right) \,{}_{0}\mathrm{D}_{t}^{\gamma
}\varepsilon \left( x,t\right) \mathrm{d}\gamma ,\;\;x\in \left[ 0,L\right]
,\;t>0,  \label{CE-dim} \\
u\left( x,0\right) =0,\;\;\;\frac{\partial }{\partial t}u\left( x,0\right)
=0,\;\;\;\sigma \left( x,0\right) =0,\;\;\;\varepsilon \left( x,0\right)
=0,\;\;x\in \left[ 0,L\right] ,  \label{IC-dim} \\
u\left( 0,t\right) =0,\;\;\;\;-A\sigma \left( L,t\right) +F\left( t\right) =m%
\frac{\partial ^{2}}{\partial t^{2}}u\left( L,t\right) ,\;\;t>0.
\label{BC-dim}
\end{gather}%
We use symbols $\sigma ,$ $u$ and $\varepsilon ,$ in the equation of motion (%
\ref{EM-SM-dim})$_{1}$ and in the strain (\ref{EM-SM-dim})$_{2},$ to denote
stress, displacement and strain, respectively, depending on the initial
position $x$ at time $t,$ while $\rho $ denotes the density of a material.
Constitutive equation (\ref{CE-dim}) corresponds to the distributed-order
fractional derivative model of a viscoelastic body, $E$ is a generalized
Young modulus (a positive constant having the dimension of a stress), $\phi
_{\sigma }$ and $\phi _{\varepsilon }$ are given constitutive functions or
distributions, ${}_{0}\mathrm{D}_{t}^{\gamma }$ is the left
Riemann-Liouville fractional derivative operator of order $\gamma \in \left(
0,1\right) $ (see \cite{SKM})%
\begin{equation*}
_{0}\mathrm{D}_{t}^{\gamma }y\left( t\right) :=\frac{\mathrm{d}}{\mathrm{d}t}%
\left( \frac{t^{-\gamma }}{\Gamma \left( 1-\gamma \right) }\ast y\left(
t\right) \right) ,\;\;t>0,
\end{equation*}%
where $\Gamma $ is the Euler gamma function and $\ast $ is the convolution.
Recall, if $f,g\in L_{loc}^{1}\left( \mathbf{%
\mathbb{R}
}\right) ,$ $\limfunc{supp}f,g\subset \left[ 0,\infty \right) ,$ then $%
\left( f\ast g\right) \left( t\right) :=\int_{0}^{t}f\left( \tau \right)
g\left( t-\tau \right) \mathrm{d}\tau ,$ $t\in
\mathbb{R}
$. We refer to \cite{TAFDE,Pod,SKM} for the basic definitions and assertions
of fractional calculus. Initial conditions in (\ref{IC-dim}) specify that
the rod-body system is unstressed and in the state of rest at the initial
time instant. Boundary condition (\ref{BC-dim})$_{1}$ means that one end of
the rod is fixed. The other boundary condition (\ref{BC-dim})$_{2}$ is the
equation of translatory motion along the $\bar{x}$ axis of the body attached
to the free end of the rod. In (\ref{BC-dim}), $A$ stands for the
cross-section area of the rod and $F$ stands for the known external force
acting on the body.

Regarding the constitutive equation (\ref{CE-dim}) we consider the following
two cases.

\begin{enumerate}
\item[I] Fractional Zener model of a viscoelastic body%
\begin{equation}
\left( 1+a\,{}_{0}\mathrm{D}_{t}^{\alpha }\right) \sigma \left( x,t\right)
=E\left( 1+b\,{}_{0}\mathrm{D}_{t}^{\alpha }\right) \varepsilon \left(
x,t\right) .  \label{Zener}
\end{equation}%
It is obtained from (\ref{CE-dim}) by choosing%
\begin{equation}
\phi _{\sigma }\left( \gamma \right) :=\delta \left( \gamma \right)
+a\,\delta \left( \gamma -\alpha \right) ,\;\;\phi _{\varepsilon }\left(
\gamma \right) :=\delta \left( \gamma \right) +b\,\delta \left( \gamma
-\alpha \right) ,\;\;\alpha \in \left( 0,1\right) ,\;0<a\leq b,
\label{zener}
\end{equation}%
where $\delta $ denotes the Dirac delta distribution.

\item[II] Distributed-order model of a viscoelastic body%
\begin{equation}
\int_{0}^{1}a^{\gamma }\,{}_{0}\mathrm{D}_{t}^{\gamma }\sigma \left(
x,t\right) \mathrm{d}\gamma =E\int_{0}^{1}b^{\gamma }\,{}_{0}\mathrm{D}%
_{t}^{\gamma }\varepsilon \left( x,t\right) \mathrm{d}\gamma .  \label{A-B}
\end{equation}%
It is obtained from (\ref{CE-dim}) by choosing%
\begin{equation}
\phi _{\sigma }\left( \gamma \right) :=a^{\gamma },\;\;\phi _{\varepsilon
}\left( \gamma \right) :=b^{\gamma },\;\;\gamma \in \left( 0,1\right)
,\;0<a\leq b.  \label{a-b}
\end{equation}%
We note that (\ref{a-b}) is the simplest form of $\phi _{\sigma }$ and $\phi
_{\varepsilon }$ providing dimensional homogeneity.
\end{enumerate}

Equation (\ref{Zener}) is often used in modeling viscoelastic bodies. It is
used in \cite{KOZ10} for the study of the wave propagation in an unbounded
domain. Equation (\ref{A-B}) is used in \cite{a-2002-a,H-L} as well as in
\cite{APZ-3,APZ-4}, where the wave motion, stress relaxation and creep, are
studied on a bounded domain. We also mention that the wave motion in a body,
described by a more general model than (\ref{Zener}) is studied in \cite%
{KOZ11}. We refer to \cite{Mai-10,R-2010,R-S-2010} for the detailed account
of applications of fractional calculus in viscoelasticity. Problems similar
to (\ref{EM-SM-dim}) - (\ref{BC-dim}) were also treated in \cite{R-S1,R-S}
with the constitutive equations related to the distributed-order model (\ref%
{CE-dim}) in the special cases.

We treated in \cite{APZ-4} the creep test of a material described by the
constitutive equation (\ref{A-B}) and concluded (according to numerical
examples) that the material is solid-like, while in \cite{AKOZ} the
constitutive function was given by%
\begin{equation}
\left( 1+a\,{}_{0}\mathrm{D}_{t}^{\alpha }\right) \sigma \left( x,t\right)
=E\left( b_{0}\,{}_{0}\mathrm{D}_{t}^{\beta _{0}}+b_{1}\,{}_{0}\mathrm{D}%
_{t}^{\beta _{1}}+b_{2}\,{}_{0}\mathrm{D}_{t}^{\beta _{2}}\right)
\varepsilon \left( x,t\right) ,  \label{Hilf}
\end{equation}%
where $a,$ $b_{0},$ $b_{1},$ $b_{2}$ are positive constants, $0<\alpha
<\beta _{0}<\beta _{1}<\beta _{2}\leq 1,$ and the conclusion was (again
based upon numerical examples) that the material is fluid-like. The
constitutive equation (\ref{Hilf}) is proposed in \cite{hilf}. In this work
we treat numerically the creep test for solid-like materials described by (%
\ref{Zener}).

\begin{remark}
Constitutive functions (or distributions) $\phi _{\sigma }$ and $\phi
_{\varepsilon }$ in (\ref{CE-dim}) have to satisfy the restrictions
following from the Second Law of Thermodynamics, see \cite%
{a-2002,a-2003,AKOZ}. We refer to \cite{AKOZ} for a systematic review of
restrictions if $\phi _{\sigma }$ and $\phi _{\varepsilon }$ are given in
the form of sums of the Dirac $\delta $ distributions
\begin{equation*}
\phi _{\sigma }\left( \gamma \right) :=\sum_{n=0}^{N}a_{n}\,\delta \left(
\gamma -\alpha _{n}\right) ,\;\;\phi _{\varepsilon }\left( \gamma \right)
:=\sum_{m=0}^{M}b_{m}\,\delta \left( \gamma -\beta _{m}\right) ,\;\;\alpha
_{n},\beta _{m}\in \left[ 0,1\right] .
\end{equation*}
\end{remark}

\begin{remark}
\label{s-f}Differences between solid and fluid-like materials are observed
in the creep test (i.e. when a material is subjected to a sudden, but later
constant force on its free end). Namely, solid-like materials creep to a
finite displacement, while the fluid-like materials creep to an infinite
displacement.
\end{remark}

The paper is organized as follows. In \S \thinspace \ref{cfs} we write the
system (\ref{EM-SM-dim}) - (\ref{BC-dim}) in the dimensionless form and
obtain (\ref{sys-1}) - (\ref{BC}). Then we formally apply the Laplace
transform to (\ref{sys-1}) - (\ref{BC}), define the function $M$ and obtain
the solutions to (\ref{sys-1}) - (\ref{BC}) in the Laplace domain via the
forcing term and solution kernel. Section \ref{P} is devoted to the
verification that the function $M$ in the cases of the fractional Zener (\ref%
{Zener}) and distributed-order model (\ref{A-B}) satisfies assumptions,
cited from \cite{APZ-5}, that imply the existence and uniqueness of the
solutions to (\ref{sys-1}) - (\ref{BC}). Then, we write theorems on
existence and uniqueness of solutions, that are proven in \cite{APZ-5}. The
explicit form of the solution $u,$ given in Theorem \ref{thmP}, is used in
\S \thinspace \ref{ne} in order to plot the solution. The plots are given
and discussed for the fractional Zener model and for two different forcing
functions.

\section{Formal solutions\label{cfs}}

We start from the system (\ref{EM-SM-dim}) - (\ref{BC-dim}) and write it in
the dimensionless form. Then, by the Laplace transform method, we obtain the
displacement $u$ and the stress $\sigma $ as the convolution of the external
force $F$ and solution kernels $P$ and $Q,$ respectively. Determination of $%
P $ and $Q$ will be given in \S \thinspace \ref{P}.

The system (\ref{EM-SM-dim}) - (\ref{BC-dim}) transforms into%
\begin{gather}
\frac{\partial }{\partial x}\sigma \left( x,t\right) =\kappa ^{2}\frac{%
\partial ^{2}}{\partial t^{2}}u\left( x,t\right) ,\;\;\;\;\varepsilon \left(
x,t\right) =\frac{\partial }{\partial x}u\left( x,t\right) ,\;\;x\in \left[
0,1\right] ,\;t>0,  \label{sys-1} \\
\int_{0}^{1}\phi _{\sigma }\left( \gamma \right) {}_{0}\mathrm{D}%
_{t}^{\gamma }\sigma \left( x,t\right) \mathrm{d}\gamma =\int_{0}^{1}\phi
_{\varepsilon }\left( \gamma \right) {}_{0}\mathrm{D}_{t}^{\gamma
}\varepsilon \left( x,t\right) \mathrm{d}\gamma ,\;\;x\in \left[ 0,1\right]
,\;t>0,  \label{sys-2} \\
u\left( x,0\right) =0,\;\;\;\frac{\partial }{\partial t}u\left( x,0\right)
=0,\;\;\;\sigma \left( x,0\right) =0,\;\;\;\varepsilon \left( x,0\right)
=0,\;\;x\in \left[ 0,1\right] ,  \label{IC} \\
u\left( 0,t\right) =0,\;\;\;\;-\sigma \left( 1,t\right) +F\left( t\right) =%
\frac{\partial ^{2}}{\partial t^{2}}u\left( 1,t\right) ,\;\;t>0.  \label{BC}
\end{gather}%
This is done by introducing the square root of the ratio between the masses
of a rod and a body%
\begin{equation*}
\kappa =\sqrt{\frac{\rho AL}{m}}
\end{equation*}%
and dimensionless quantities%
\begin{equation*}
\bar{x}=\frac{x}{L},\;\bar{t}=\frac{t}{\sqrt{\frac{mL}{AE}}},\;\bar{u}=\frac{%
u}{L},\;\bar{\sigma}=\frac{\sigma }{E},\;\bar{\phi}_{\sigma }=\frac{\phi
_{\sigma }}{\left( \sqrt{\frac{mL}{AE}}\right) ^{\gamma }},\;\bar{\phi}%
_{\varepsilon }=\frac{\phi _{\varepsilon }}{\left( \sqrt{\frac{mL}{AE}}%
\right) ^{\gamma }},\;\bar{F}=\frac{F}{AE}.
\end{equation*}%
In writing (\ref{sys-1}) - (\ref{BC}) we omitted bar over dimensionless
quantities. Note that the choice of dimensionless quantities implies that
the case of a rod without the attached mass ($m=0$) cannot be studied as a
special case of equations (\ref{sys-1}) - (\ref{BC}).

In order to solve the system (\ref{sys-1}), (\ref{sys-2}) subjected to the
initial (\ref{IC}) and boundary data (\ref{BC}), we use the Laplace
transform method. Recall that the Laplace transform of $f\in
L_{loc}^{1}\left( \mathbf{%
\mathbb{R}
}\right) ,$ $f\equiv 0$ in $\left( -\infty ,0\right] $ and $\left\vert
f\left( t\right) \right\vert \leq c\mathrm{e}^{kt},$ $t>0,$ for some $k>0,$
is defined by%
\begin{equation*}
\tilde{f}\left( s\right) =\mathcal{L}\left[ f\left( t\right) \right] \left(
s\right) :=\int_{0}^{\infty }f\left( t\right) e^{-st}\mathrm{d}t,\;\;\func{Re%
}s>k
\end{equation*}%
and analytically continued into the appropriate domain $D.$ Moreover, we
consider our problems within the the space of tempered generalized functions
supported by $\left[ 0,\infty \right) ,$ denoted by $\mathcal{S}_{+}^{\prime
}.$ The Laplace transform within this space is an extension of the classical
one, given above. Namely, any $g\in \mathcal{S}_{+}^{\prime }$ of the form $%
g:={}_{0}\mathrm{D}_{t}^{\gamma }f,$ $\gamma \in \left( 0,1\right) ,$ where $%
f$ is as above (polynomially bounded) satisfies $\mathcal{L}\left[ g\left(
t\right) \right] \left( s\right) =s^{\gamma }\tilde{f}\left( s\right) ,$ $%
\func{Re}s>0.$ We refer to \cite{vlad} for the properties of elements of $%
\mathcal{S}_{+}^{\prime }$ and their Laplace transforms.

Applying formally the Laplace transform to (\ref{sys-1}) - (\ref{BC}) we
obtain%
\begin{gather}
\frac{\partial }{\partial x}\tilde{\sigma}\left( x,s\right) =\kappa ^{2}s^{2}%
\tilde{u}\left( x,s\right) ,\;\;\;\;\tilde{\varepsilon}\left( x,s\right) =%
\frac{\partial }{\partial x}\tilde{u}\left( x,s\right) ,\;\;x\in \left[ 0,1%
\right] ,\;s\in D,  \label{S-LT-1} \\
\tilde{\sigma}\left( x,s\right) \int_{0}^{1}\phi _{\sigma }\left( \gamma
\right) s^{\gamma }\mathrm{d}\gamma =\tilde{\varepsilon}\left( x,s\right)
\int_{0}^{1}\phi _{\varepsilon }\left( \gamma \right) s^{\gamma }\mathrm{d}%
\gamma ,\;\;x\in \left[ 0,1\right] ,\;s\in D,  \label{S-LT-2} \\
\tilde{u}\left( 0,s\right) =0,\;\;\tilde{\sigma}\left( 1,s\right) +s^{2}%
\tilde{u}\left( 1,s\right) =\tilde{F}\left( s\right) ,\;\;s\in D.
\label{S-LT-3}
\end{gather}%
By (\ref{S-LT-2}) we have
\begin{equation}
\tilde{\sigma}\left( x,s\right) =\frac{1}{M^{2}\left( s\right) }\tilde{%
\varepsilon}\left( x,s\right) ,\;\;s\in D,  \label{sigma-tilda}
\end{equation}%
where we introduced%
\begin{equation}
M\left( s\right) :=\sqrt{\frac{\int_{0}^{1}\phi _{\sigma }\left( \gamma
\right) s^{\gamma }\mathrm{d}\gamma }{\int_{0}^{1}\phi _{\varepsilon }\left(
\gamma \right) s^{\gamma }\mathrm{d}\gamma }},\;\;s\in D.  \label{M}
\end{equation}%
Thus, using (\ref{zener}) and (\ref{a-b}), in the cases of constitutive
equations (\ref{Zener}) and (\ref{A-B}) we obtain%
\begin{eqnarray}
M\left( s\right)  &=&\sqrt{\frac{1+as^{\alpha }}{1+bs^{\alpha }}},\;\;s\in
\mathbb{C}
\backslash \left( -\infty ,0\right] ,\;0<a\leq b,\;\alpha \in \left(
0,1\right) ,  \label{zener-M} \\
M\left( s\right)  &=&\sqrt{\frac{\ln \left( bs\right) }{\ln \left( as\right)
}\frac{as-1}{bs-1}},\;\;s\in
\mathbb{C}
\backslash \left( -\infty ,0\right] ,\;0<a\leq b.  \label{a-b-M}
\end{eqnarray}%
So, in the sequel $D=%
\mathbb{C}
\backslash \left( -\infty ,0\right] .$

In order to obtain the displacement $u,$ we use (\ref{S-LT-1}), (\ref%
{sigma-tilda}) and obtain%
\begin{equation}
\frac{\partial ^{2}}{\partial x^{2}}\tilde{u}\left( x,s\right) -\kappa
^{2}s^{2}M^{2}\left( s\right) \tilde{u}\left( x,s\right) =0,\;\;x\in \left[
0,1\right] ,\;s\in
\mathbb{C}
\backslash \left( -\infty ,0\right] ,  \label{eq}
\end{equation}%
The solution of (\ref{eq}) is%
\begin{equation*}
\tilde{u}\left( x,s\right) =C_{1}\left( s\right) \mathrm{e}^{\kappa
xsM\left( s\right) }+C_{2}\left( s\right) \mathrm{e}^{-\kappa xsM\left(
s\right) },\;\;x\in \left[ 0,1\right] ,\;s\in
\mathbb{C}
\backslash \left( -\infty ,0\right] ;
\end{equation*}%
$C_{1}$ and $C_{2}$ are arbitrary functions which are determined from (\ref%
{S-LT-3})$_{1}$ as $2C=C_{1}=-C_{2}.$ Therefore,%
\begin{equation}
\tilde{u}\left( x,s\right) =C\left( s\right) \sinh \left( \kappa xsM\left(
s\right) \right) ,\;\;x\in \left[ 0,1\right] ,\;s\in
\mathbb{C}
\backslash \left( -\infty ,0\right] .  \label{u-tilda-1}
\end{equation}%
By (\ref{S-LT-1})$_{2}$, (\ref{sigma-tilda}) and (\ref{u-tilda-1}) we have
\begin{equation}
\tilde{\sigma}\left( x,s\right) =C\left( s\right) \frac{\kappa s}{M\left(
s\right) }\cosh \left( \kappa xsM\left( s\right) \right) ,\;\;x\in \left[ 0,1%
\right] ,\;s\in
\mathbb{C}
\backslash \left( -\infty ,0\right] .  \label{sigma-tilda-1}
\end{equation}%
Using (\ref{u-tilda-1}) and (\ref{sigma-tilda-1}) at $x=1,$ as well as (\ref%
{S-LT-3})$_{2}$ we obtain%
\begin{equation*}
C\left( s\right) =\frac{M\left( s\right) \tilde{F}\left( s\right) }{s\left(
sM\left( s\right) \sinh \left( \kappa sM\left( s\right) \right) +\kappa
\cosh \left( \kappa sM\left( s\right) \right) \right) },\;\;s\in
\mathbb{C}
\backslash \left( -\infty ,0\right] .
\end{equation*}%
Therefore, the Laplace transforms of the displacement and stress from (\ref%
{u-tilda-1}) and (\ref{sigma-tilda-1}) are
\begin{equation}
\tilde{u}\left( x,s\right) =\tilde{F}\left( s\right) \tilde{P}\left(
x,s\right) ,\;\;\tilde{\sigma}\left( x,s\right) =\tilde{F}\left( s\right)
\tilde{Q}\left( x,s\right) ,\;\;x\in \left[ 0,1\right] ,\;s\in
\mathbb{C}
\backslash \left( -\infty ,0\right] .  \label{u,sigma-tilda}
\end{equation}%
where%
\begin{eqnarray}
\tilde{P}\left( x,s\right) &=&\frac{1}{s}\frac{M\left( s\right) \sinh \left(
\kappa xsM\left( s\right) \right) }{sM\left( s\right) \sinh \left( \kappa
sM\left( s\right) \right) +\kappa \cosh \left( \kappa sM\left( s\right)
\right) },\;\;x\in \left[ 0,1\right] ,\;s\in
\mathbb{C}
\backslash \left( -\infty ,0\right] ,  \label{P-tilda} \\
\tilde{Q}\left( x,s\right) &=&\frac{\kappa \cosh \left( \kappa xsM\left(
s\right) \right) }{sM\left( s\right) \sinh \left( \kappa sM\left( s\right)
\right) +\kappa \cosh \left( \kappa sM\left( s\right) \right) },\;\;x\in %
\left[ 0,1\right] ,\;s\in
\mathbb{C}
\backslash \left( -\infty ,0\right] .  \label{Q-tilda}
\end{eqnarray}%
Applying the inverse Laplace transform to (\ref{u,sigma-tilda}) we obtain
the displacement and stress as%
\begin{equation}
u\left( x,t\right) =F\left( t\right) \ast P\left( x,t\right) ,\;\;\sigma
\left( x,t\right) =F\left( t\right) \ast Q\left( x,t\right) ,\;\;x\in \left[
0,1\right] ,\;t>0.  \label{u, sigma}
\end{equation}%
The validity of these formal expressions will be proved in the sequel.

\section{Explicit form of solutions\label{P}}

In order to obtain the displacement $u$ and stress $\sigma $ by (\ref{u,
sigma}), we have to obtain functions $P$ and $Q,$ i.e., to invert the
Laplace transform in (\ref{P-tilda}) and (\ref{Q-tilda}). First, we examine
the behavior of the function $M,$ given by (\ref{M}), in the limiting cases
as $\left\vert s\right\vert \rightarrow \infty $ and $\left\vert
s\right\vert \rightarrow 0$ (by this we mean only those $s$ that belong to $%
\mathbb{C}
\backslash \left( -\infty ,0\right] $) in the special cases when $M$ takes
any of the forms given by (\ref{zener-M}) and (\ref{a-b-M}).

If $M$ is given by (\ref{zener-M}) or (\ref{a-b-M}) we have%
\begin{equation}
\left\vert M\left( s\right) \right\vert \approx \sqrt{\frac{a}{b}},\;\;\text{%
as}\;\;\left\vert s\right\vert \rightarrow \infty ,\;\;\text{and}%
\;\;\left\vert M\left( s\right) \right\vert \approx 1,\;\;\text{as}%
\;\;\left\vert s\right\vert \rightarrow 0.  \label{s-l}
\end{equation}

We use results from \cite{APZ-5} in order to obtain the displacement $u$ and
the stress $\sigma .$ In order to do so, we recall assumptions on $M$ that
have to be satisfied.

We shall analyze a function of complex variable%
\begin{equation}
f\left( s\right) :=sM\left( s\right) \sinh \left( \kappa sM\left( s\right)
\right) +\kappa \cosh \left( \kappa sM\left( s\right) \right) ,\;\;s\in
\mathbb{C}
.  \label{polovi-01}
\end{equation}

Let $M$ be of the form $M\left( s\right) =r\left( s\right) +\mathrm{i}%
h\left( s\right) ,$ as $\left\vert s\right\vert \rightarrow \infty .$ Then

\begin{enumerate}
\item[(A1)]
\begin{eqnarray*}
&&\lim_{\left\vert s\right\vert \rightarrow \infty }r\left( s\right)
=c_{\infty }>0,\;\;\lim_{\left\vert s\right\vert \rightarrow \infty }h\left(
s\right) =0,\;\;\lim_{\left\vert s\right\vert \rightarrow 0}M\left( s\right)
=c_{0}, \\
&&\text{for some constants}\;c_{\infty },c_{0}>0.
\end{eqnarray*}
\end{enumerate}

Let $s_{n}=\xi _{n}+\mathrm{i}\zeta _{n},$ $n\in
\mathbb{N}
,$ satisfy the equation%
\begin{equation}
f\left( s\right) =0,\;\;s\in V,  \label{polovi-0}
\end{equation}%
where $f$ is given by (\ref{polovi-01}). Then:

\begin{enumerate}
\item[(A2)] There exists $n_{0}>0,$ such that for $n>n_{0}$%
\begin{eqnarray*}
&&\func{Im}s_{n}\in
\mathbb{R}
_{+}\Rightarrow h\left( s_{n}\right) \leq 0,\;\;\;\;\func{Im}s_{n}\in
\mathbb{R}
_{-}\Rightarrow h\left( s_{n}\right) \geq 0, \\
&&\text{where }h:=\func{Im}M.
\end{eqnarray*}

\item[(A3)] There exist $s_{0}>0$ and $c>0$ such that
\begin{equation*}
\left\vert \frac{\mathrm{d}}{\mathrm{d}s}(sM\left( s\right) )\right\vert
\geq c,\;\;\left\vert s\right\vert >s_{0}.
\end{equation*}

\item[(A4)] For every $\gamma >0$ there exists $\theta >0$ and $s_{0}$ such
that%
\begin{equation*}
\left\vert \left( s+\Delta s\right) M\left( s+\Delta s\right) -sM\left(
s\right) \right\vert \leq \gamma ,\;\;\text{if}\;\;\left\vert \Delta
s\right\vert <\theta \;\;\text{and}\;\;\left\vert s\right\vert >s_{0}.
\end{equation*}
\end{enumerate}

In order to write $M=r+\mathrm{i}h,$ as required above, we start from $%
M^{2}=u+\mathrm{i}v$ and obtain the system
\begin{equation*}
r^{2}-h^{2}=u,\;\;\;\;2rh=v.
\end{equation*}%
Solutions of the previous system belonging to the set of real numbers are%
\begin{eqnarray}
r &=&\pm \frac{\sqrt{2}}{2}\sqrt{\sqrt{u^{2}+v^{2}}+u},  \label{r} \\
h &=&\pm \frac{\sqrt{2}}{2}\sqrt{\sqrt{u^{2}+v^{2}}-u}.  \label{h}
\end{eqnarray}%
Assume $\frac{v}{u}\rightarrow 0,$ $u>0.$ Then by using the approximation $%
\left( 1+x\right) ^{a}\approx 1+ax,$ as $x\rightarrow 0^{+},$ from (\ref{r})
and (\ref{h}), we have%
\begin{equation}
r\approx \pm \sqrt{u},\;\;\;\;h\approx \pm \frac{1}{2}\frac{\left\vert
v\right\vert }{\sqrt{u}}.  \label{r,h}
\end{equation}

\begin{proposition}
Functions $M$ given by (\ref{zener-M}) and (\ref{a-b-M}) satisfy $\left(
\mathrm{A1}\right) $ - $\left( \mathrm{A4}\right) .$
\end{proposition}

\begin{proof}
Consider $M$ given by (\ref{zener-M}). If we write $s=r\mathrm{e}^{\mathrm{i}%
\varphi }$, by (\ref{zener-M}), we have%
\begin{eqnarray*}
M^{2}\left( s\right) &=&\frac{1+as^{\alpha }}{1+bs^{\alpha }},\;\;a\leq b \\
&=&\frac{1+\left( a+b\right) r^{\alpha }\cos \left( \alpha \varphi \right)
+abr^{2\alpha }}{1+2br^{\alpha }\cos \left( \alpha \varphi \right) +\left(
br^{\alpha }\right) ^{2}}-\mathrm{i}\frac{\left( b-a\right) r^{\alpha }\sin
\left( \alpha \varphi \right) }{1+2br^{\alpha }\cos \left( \alpha \varphi
\right) +\left( br^{\alpha }\right) ^{2}} \\
&\approx &\frac{a}{b}-\mathrm{i}\frac{b-a}{b^{2}}\frac{\sin \left( \alpha
\varphi \right) }{r^{\alpha }},\;\;\text{as}\;\;\left\vert s\right\vert
\rightarrow \infty .
\end{eqnarray*}%
Using (\ref{r,h}) we obtain%
\begin{equation*}
r\left( s\right) \approx \pm \sqrt{\frac{a}{b}},\;\;\;\;h\left( s\right)
\approx \pm \frac{1}{2}\sqrt{\frac{b}{a}}\frac{b-a}{b^{2}}\frac{\left\vert
\sin \left( \alpha \varphi \right) \right\vert }{r^{\alpha }},\;\;\text{as}%
\;\;\left\vert s\right\vert \rightarrow \infty .
\end{equation*}%
Let $\varphi \in \left( 0,\pi \right) $, then $\sin \left( \alpha \varphi
\right) >0.$ This implies $\func{Re}\left( M^{2}\left( s\right) \right) >0$
and $\func{Im}\left( M^{2}\left( s\right) \right) <0.$ Therefore, we also
have $\func{Re}\left( M\left( s\right) \right) >0$ and $\func{Im}\left(
M\left( s\right) \right) <0.$ Similar arguments are valid if $\varphi \in
\left( -\pi ,0\right) .$ Hence, we finally have%
\begin{equation}
r\left( s\right) \approx \sqrt{\frac{a}{b}},\;\;\;\;h\left( s\right) \approx
-\frac{1}{2}\sqrt{\frac{b}{a}}\frac{b-a}{b^{2}}\frac{\sin \left( \alpha
\varphi \right) }{r^{\alpha }},\;\;\text{as}\;\;\left\vert s\right\vert
\rightarrow \infty .  \label{r,h-zen}
\end{equation}

Next we prove that $M$ in (\ref{zener-M}) satisfies $\left( \mathrm{A1}%
\right) .$ By (\ref{r,h-zen}) and (\ref{zener-M}), we have
\begin{equation*}
\lim_{\left\vert s\right\vert \rightarrow \infty }r\left( s\right) =\sqrt{%
\frac{a}{b}},\;\;\lim_{\left\vert s\right\vert \rightarrow \infty }h\left(
s\right) =0,\;\;\lim_{\left\vert s\right\vert \rightarrow 0}\left\vert
M\left( s\right) \right\vert =1.
\end{equation*}

Validity of assumption $\left( \mathrm{A2}\right) $ follows from (\ref%
{r,h-zen}).

In order to show that $M$ in (\ref{zener-M}) satisfies $\left( \mathrm{A3}%
\right) ,$ we use (\ref{zener-M}) and obtain
\begin{equation*}
\frac{\mathrm{d}}{\mathrm{d}s}\left( sM\left( s\right) \right) =M\left(
s\right) \left( 1-\frac{\alpha \left( b-a\right) s^{\alpha }}{2\left(
1+as^{\alpha }\right) \left( 1+bs^{\alpha }\right) }\right) .
\end{equation*}%
Thus, by (\ref{r,h-zen})%
\begin{equation}
\left\vert \frac{\mathrm{d}}{\mathrm{d}s}\left( sM\left( s\right) \right)
\right\vert \approx \sqrt{\frac{a}{b}},\;\;\text{as}\;\;\left\vert
s\right\vert \rightarrow \infty .  \label{izvod-sms-1}
\end{equation}

We have that there exists $\xi $ such that%
\begin{equation*}
\left\vert \left( s+\Delta s\right) M\left( s+\Delta s\right) -sM\left(
s\right) \right\vert \leq \left\vert \Delta s\right\vert \left\vert \left[
\frac{\mathrm{d}}{\mathrm{d}s}\left( sM\left( s\right) \right) \right]
_{s=\xi }\right\vert .
\end{equation*}%
Since $\frac{\mathrm{d}}{\mathrm{d}s}\left( sM\left( s\right) \right) ,$ by (%
\ref{izvod-sms-1}), is bounded as $\left\vert s\right\vert \rightarrow
\infty $ and if $\left\vert \Delta s\right\vert <\theta $ for some $\theta
>0,$ we have that $\left( \mathrm{A4}\right) $ is satisfied.

Now consider $M$ given by (\ref{a-b-M}). With $s=r\mathrm{e}^{\mathrm{i}%
\varphi }$ we have%
\begin{eqnarray*}
M^{2}\left( s\right) &=&\frac{\ln \left( bs\right) }{\ln \left( as\right) }%
\frac{as-1}{bs-1},\;\;a\leq b \\
&=&\frac{\ln \left( ar\right) \ln \left( br\right) +\varphi ^{2}}{\ln
^{2}\left( ar\right) +\varphi ^{2}}\frac{abr^{2}-\left( a+b\right) r\cos
\varphi +1}{\left( br\right) ^{2}-2br\cos \varphi +1}-\frac{\varphi \ln
\frac{b}{a}}{\ln ^{2}\left( ar\right) +\varphi ^{2}}\frac{\left( b-a\right)
r\sin \varphi }{\left( br\right) ^{2}-2br\cos \varphi +1} \\
&&-\mathrm{i}\left( \frac{\ln \left( ar\right) \ln \left( br\right) +\varphi
^{2}}{\ln ^{2}\left( ar\right) +\varphi ^{2}}\frac{\left( b-a\right) r\sin
\varphi }{\left( br\right) ^{2}-2br\cos \varphi +1}+\frac{\varphi \ln \frac{b%
}{a}}{\ln ^{2}\left( ar\right) +\varphi ^{2}}\frac{abr^{2}-\left( a+b\right)
r\cos \varphi +1}{\left( br\right) ^{2}-2br\cos \varphi +1}\right) \\
&\approx &\frac{a}{b}-\mathrm{i}\frac{a}{b}\ln \frac{b}{a}\frac{\varphi }{%
\ln ^{2}\left( ar\right) },\;\;\text{as}\;\;\left\vert s\right\vert
\rightarrow \infty .
\end{eqnarray*}%
Using (\ref{r,h}) we obtain%
\begin{equation*}
r\left( s\right) =\pm \sqrt{\frac{a}{b}},\;\;\;\;h\left( s\right) =\pm \frac{%
1}{2}\sqrt{\frac{a}{b}}\ln \frac{b}{a}\frac{\left\vert \varphi \right\vert }{%
\ln ^{2}\left( ar\right) }.
\end{equation*}%
Let $\varphi \in \left( 0,\pi \right) $, then $\sin \left( \alpha \varphi
\right) >0.$ This implies $\func{Re}\left( M^{2}\left( s\right) \right) >0$
and $\func{Im}\left( M^{2}\left( s\right) \right) <0.$ Therefore, we also
have $\func{Re}\left( M\left( s\right) \right) >0$ and $\func{Im}\left(
M\left( s\right) \right) <0.$ Similar arguments are valid if $\varphi \in
\left( -\pi ,0\right) .$ Hence, we finally have%
\begin{equation}
r\left( s\right) =\sqrt{\frac{a}{b}},\;\;\;\;h\left( s\right) =-\frac{1}{2}%
\sqrt{\frac{a}{b}}\ln \frac{b}{a}\frac{\varphi }{\ln ^{2}\left( ar\right) }.
\label{r,h-a-b}
\end{equation}

Next we prove that (\ref{a-b-M}) satisfies $\left( \mathrm{A1}\right) .$ By (%
\ref{r,h-a-b}) and (\ref{a-b-M}), we have
\begin{equation*}
\lim_{\left\vert s\right\vert \rightarrow \infty }r\left( s\right) =\sqrt{%
\frac{a}{b}},\;\;\lim_{\left\vert s\right\vert \rightarrow \infty }h\left(
s\right) =0,\;\;\lim_{\left\vert s\right\vert \rightarrow 0}\left\vert
M\left( s\right) \right\vert =1.
\end{equation*}

Validity of assumption $\left( \mathrm{A2}\right) $ follows from (\ref%
{r,h-a-b}).

In order to show $M$ in (\ref{a-b-M}) satisfies $\left( \mathrm{A3}\right) ,$
we use (\ref{a-b-M}) and obtain
\begin{equation*}
\frac{\mathrm{d}}{\mathrm{d}s}\left( sM\left( s\right) \right) =M\left(
s\right) \left( 1-\frac{\ln \frac{b}{a}}{2\ln \left( as\right) \ln \left(
bs\right) }+\frac{\left( b-a\right) s}{2\left( as-1\right) \left(
bs-1\right) }\right) .
\end{equation*}%
Thus, by (\ref{r,h-a-b})%
\begin{equation}
\left\vert \frac{\mathrm{d}}{\mathrm{d}s}\left( sM\left( s\right) \right)
\right\vert \approx \sqrt{\frac{a}{b}},\;\;\text{as}\;\;\left\vert
s\right\vert \rightarrow \infty .  \label{izvod-sms-2}
\end{equation}

Using the same arguments as above, we have that $\left( \mathrm{A4}\right) $
is satisfied, since $sM\left( s\right) ,$ $s\in V,$ by (\ref{izvod-sms-2}),
has bounded first derivative.
\end{proof}

The existence and the uniqueness of $u$ and $\sigma ,$ as solutions to
system (\ref{sys-1}) - (\ref{BC}) is guaranteed by the fact that $M$ in all
two cases satisfies $\left( \mathrm{A1}\right) $ - $\left( \mathrm{A4}%
\right) .$ Recall, $f$ is given by (\ref{polovi-01}) and $s_{n},$ $n\in
\mathbb{N}
,$ are solutions of (\ref{polovi-0}). The multiplicity of zeros $s_{n}$ is
one for $n$ large enough.

\begin{theorem}[\protect\cite{APZ-5}]
\label{thmP}Let $F\in \mathcal{S}_{+}^{\prime }$ and suppose that $M$
satisfies assumptions $\left( \mathrm{A1}\right) $ - $\left( \mathrm{A4}%
\right) .$ Then the unique solution $u$ to (\ref{sys-1}) - (\ref{BC}) is
given by%
\begin{equation*}
u\left( x,t\right) =F\left( t\right) \ast P\left( x,t\right) ,\;\;x\in \left[
0,1\right] ,\;t>0,
\end{equation*}%
where%
\begin{eqnarray*}
P\left( x,t\right) &=&\frac{1}{\pi }\dint\nolimits_{0}^{\infty }\func{Im}%
\left( \frac{M\left( q\mathrm{e}^{-\mathrm{i}\pi }\right) \sinh \left(
\kappa xqM\left( q\mathrm{e}^{-\mathrm{i}\pi }\right) \right) }{qM\left( q%
\mathrm{e}^{-\mathrm{i}\pi }\right) \sinh \left( \kappa qM\left( q\mathrm{e}%
^{-\mathrm{i}\pi }\right) \right) +\kappa \cosh \left( \kappa qM\left( q%
\mathrm{e}^{-\mathrm{i}\pi }\right) \right) }\right) \frac{\mathrm{e}^{-qt}}{%
q}\mathrm{d}q \\
&&+2\sum_{n=1}^{\infty }\func{Re}\left( \func{Res}\left( \tilde{P}\left(
x,s\right) \mathrm{e}^{st},s_{n}\right) \right) ,\;\;x\in \left[ 0,1\right]
,\;t>0, \\
P\left( x,t\right) &=&0,\;\;x\in \left[ 0,1\right] ,\;t<0.
\end{eqnarray*}%
The residues are given by%
\begin{equation*}
\func{Res}\left( \tilde{P}\left( x,s\right) \mathrm{e}^{st},s_{n}\right) =%
\left[ \frac{1}{s}\frac{M\left( s\right) \sinh \left( \kappa xsM\left(
s\right) \right) }{\frac{\mathrm{d}}{\mathrm{d}s}f\left( s\right) }\mathrm{e}%
^{st}\right] _{s=s_{n}},\;\;x\in \left[ 0,1\right] ,\;t>0,
\end{equation*}

Then $P\in C\left( \left[ 0,1\right] \times \left[ 0,\infty \right) \right) $
and $u\in C\left( \left[ 0,1\right] ,\mathcal{S}_{+}^{\prime }\right) .$ In
particular, if $F\in L_{loc}^{1}\left( \left[ 0,\infty \right) \right) ,$
then $u$ is continuous on $\left[ 0,1\right] \times \left[ 0,\infty \right)
. $
\end{theorem}

The following theorem is related to stress $\sigma .$ We formulate this
theorem with $F=H,$ where $H$ denotes the Heaviside function, while the more
general cases of $F$ are discussed in Remark \ref{o-efu}, below.

\begin{theorem}[\protect\cite{APZ-5}]
\label{thmQ}Let $F=H$ and suppose that $M$ satisfies assumptions $\left(
\mathrm{A1}\right) $ - $\left( \mathrm{A4}\right) .$ Then the unique
solution $\sigma _{H}$ to (\ref{sys-1}) - (\ref{BC}), is given by%
\begin{eqnarray}
\sigma _{H}\left( x,t\right) &=&H\left( t\right) +\frac{\kappa }{\pi }%
\dint\nolimits_{0}^{\infty }\func{Im}\left( \frac{\cosh \left( \kappa
xqM\left( q\mathrm{e}^{\mathrm{i}\pi }\right) \right) }{qM\left( q\mathrm{e}%
^{\mathrm{i}\pi }\right) \sinh \left( \kappa qM\left( q\mathrm{e}^{\mathrm{i}%
\pi }\right) \right) +\kappa \cosh \left( \kappa qM\left( q\mathrm{e}^{%
\mathrm{i}\pi }\right) \right) }\right) \frac{\mathrm{e}^{-qt}}{q}\mathrm{d}q
\notag \\
&&+2\sum_{n=1}^{\infty }\func{Re}\left( \func{Res}\left( \tilde{\sigma}%
_{H}\left( x,s\right) \mathrm{e}^{st},s_{n}\right) \right) ,\;\;x\in \left[
0,1\right] ,\;t>0,  \label{Q1} \\
\sigma _{H}\left( x,t\right) &=&0,\;\;x\in \left[ 0,1\right] ,\;t<0.  \notag
\label{Q0}
\end{eqnarray}%
The residues are given by%
\begin{equation*}
\func{Res}\left( \sigma _{H}\left( x,s\right) \mathrm{e}^{st},s_{n}\right) =%
\left[ \frac{\kappa \cosh \left( \kappa xsM\left( s\right) \right) }{s\frac{%
\mathrm{d}}{\mathrm{d}s}f\left( s\right) }\mathrm{e}^{st}\right]
_{s=s_{n}},\;\;x\in \left[ 0,1\right] ,\;t>0.
\end{equation*}%
In particular, $\sigma _{H}$ is continuous on $\left[ 0,1\right] \times %
\left[ 0,\infty \right) .$
\end{theorem}

\begin{remark}[\protect\cite{APZ-5}]
\label{o-efu}\qquad

\begin{enumerate}
\item The assumption $F=H$ in Theorem \ref{thmQ} can be relaxed by requiring
that $F$ is locally integrable and
\begin{equation*}
\tilde{F}\left( s\right) \approx \frac{1}{s^{\alpha }},\;\;\text{as}%
\;\;\left\vert s\right\vert \rightarrow \infty ,
\end{equation*}%
for some $\alpha \in \left( 0,1\right) .$ This condition ensures the
convergence of the series in (\ref{Q1}).

\item If $F=\delta ,$ or even $F\left( t\right) =\frac{\mathrm{d}^{k}}{%
\mathrm{d}t^{k}}\delta \left( t\right) ,$ one uses $\sigma _{H},$ given by (%
\ref{Q1}), in order to obtain $\sigma $ as the $k+1$-th distributional
derivative:%
\begin{equation*}
\sigma =\frac{\mathrm{d}^{k+1}}{\mathrm{d}t^{k+1}}\sigma _{H}\in C\left( %
\left[ 0,1\right] ,\mathcal{S}_{+}^{\prime }\right) .
\end{equation*}
\end{enumerate}
\end{remark}

\section{Numerical examples \label{ne}}

The displacement $u$ as a solution to (\ref{sys-1}) - (\ref{BC}) is given in
Theorem \ref{thmP}. We present various numerical examples for constitutive
models: fractional Zener and distributed-order model of a solid-like
viscoelastic body that are distinguished by the form of $M$: (\ref{zener-M})
and (\ref{a-b-M}), respectively.

\subsection{The case $F=\protect\delta $}

In order to plot time dependence of the displacement $u$ for the several
points of the rod as well as for the body attached to its free end, we chose
the fractional Zener model and the force acting on the body to be the Dirac
delta distribution, i.e. $F=\delta .$ We fix the parameters describing the
rod: $a=0.2,$ $b=0.6,$ $\alpha =0.45$ and also fix the ratio between the
masses of rod and body $\varkappa :=\kappa ^{2}=1.$ Plot of $u$ as a
function of time $t\ $for various points of a rod is shown in Figure \ref%
{fig-2}.
\begin{figure}[h]
\begin{center}
\includegraphics[scale=1]{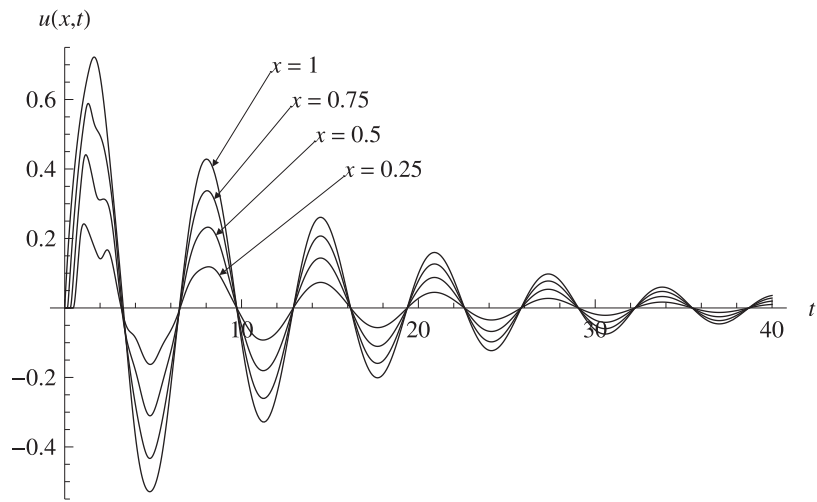}
\end{center}
\caption{Displacement $u(x,t)$ in the case $F=\protect\delta $ for $%
\varkappa =1$ as a function of time $t$ at $x\in \{0.25,0.5,0.75,1\}$ for $%
t\in (0,40)$.}
\label{fig-2}
\end{figure}
It is evident that the oscillations of the rod and a body are damped, since
the material is viscoelastic. One notices that initially there is a
transitional regime of the oscillations. Afterwards, the curves resemble the
curves of the damped linear oscillator.

In order to examine the transitional regime more closely, in Figures \ref%
{fig-4} - \ref{fig-6} we present the plots of $u$ for smaller values of
time, but for different values of $\varkappa \in \left\{ 0.5,1,2\right\} $.
\begin{figure}[h]
\begin{minipage}{72mm}
 \centering
 \includegraphics[scale=0.7]{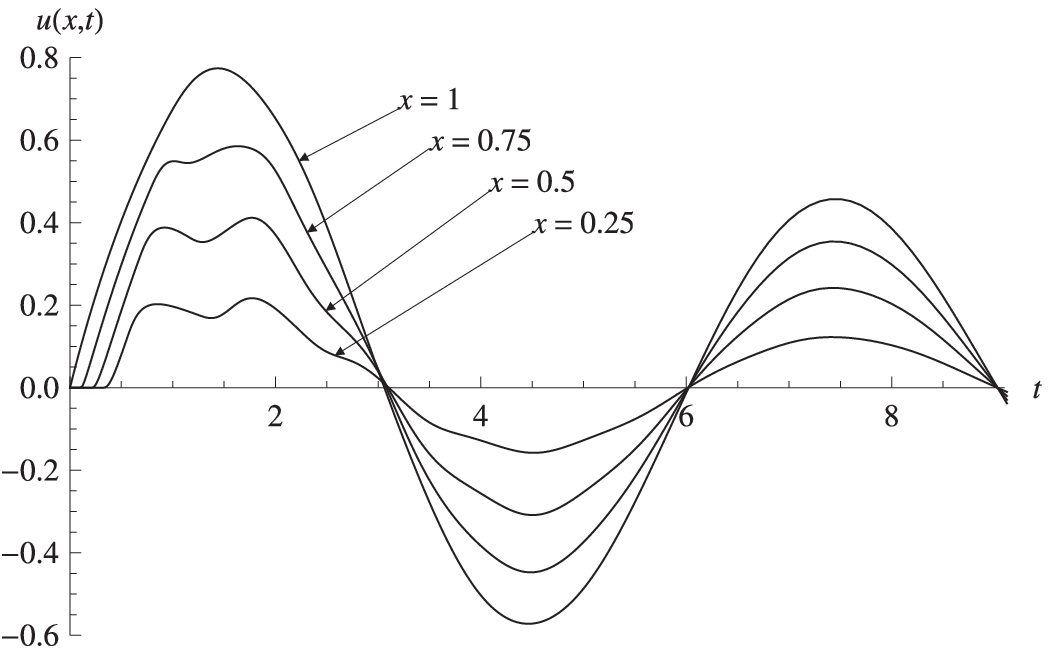}
 \caption{Displacement
 $u(x,t)$ in the case $F=\delta$ for $\varkappa=0.5$ as a function of time $t$ at
 $x\in \{0.25,0.5,0.75,1\}$ for $t\in (0,9)$.}
 \label{fig-4}
 \end{minipage}
\hfil
\begin{minipage}{72mm}
 \centering
 \includegraphics[scale=0.7]{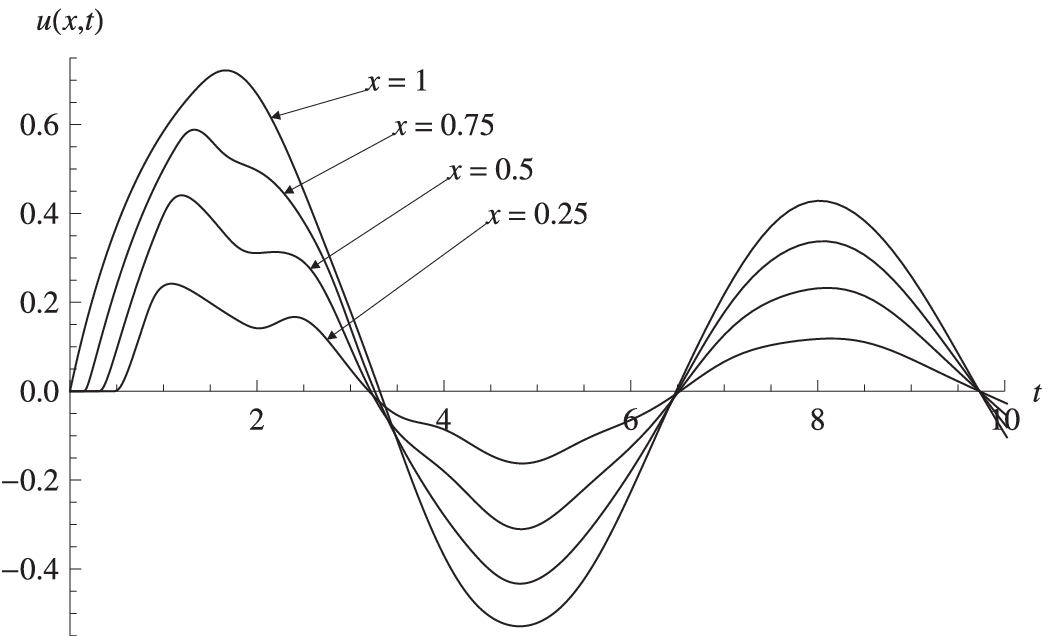}
 \caption{Displacement
 $u(x,t)$ in the case $F=\delta$ for $\varkappa=1$ as a function of time $t$ at
 $x\in \{0.25,0.5,0.75,1\}$ for $t\in (0,10)$.}
 \label{fig-5}
 \end{minipage}
\hfil
\par
\begin{center}
\begin{minipage}{72mm}
 \centering
 \includegraphics[scale=0.7]{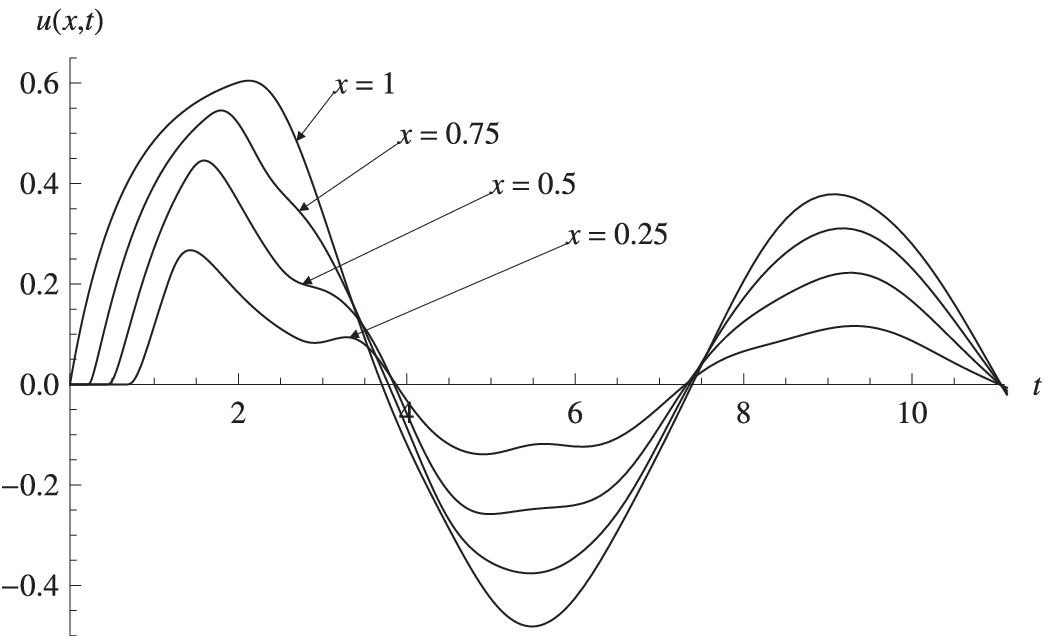}
 \caption{Displacement
 $u(x,t)$ in the case $F=\delta$ for $\varkappa=2$ as a function of time $t$ at
 $x\in \{0.25,0.5,0.75,1\}$ for $t\in (0,11)$.}
 \label{fig-6}
 \end{minipage}
\end{center}
\end{figure}
We notice that the shape of the curves depends on the ratio between the
masses $\varkappa ,$ while later, the shape resembles to the shape of curves
for damped oscillations. It could be noticed that regardless of the value of
$\varkappa $ there is a delay in starting oscillation for the points that
are further away from the free end of a rod. This is due to the finite speed
of wave propagation. Namely, the body ($x=1$), which is subject to the
action of the force, starts oscillating at $t=0,$ while the delay in the
starting time-instant of the oscillation is greater as the point is further
from the point where the force acts. Moreover, we see that different points
of the rod do not come to their initial position, for the first time, at the
same time-instant (see Figure \ref{fig-6}). This, as well as the initial
delay depend on the mass ratio $\varkappa .$ For the influence of $\varkappa
$ on the initial delay compare Figures \ref{fig-4} - \ref{fig-6}. Later on,
again depending on $\varkappa ,$ the motion of the points become
synchronized.

Figures \ref{fig-7} and \ref{fig-8} present the plots of displacement $u$
for fixed point of the rod $x=0.5$ if the ratio between masses $\varkappa $
varies.
\begin{figure}[h]
\centering
\includegraphics[scale=1]{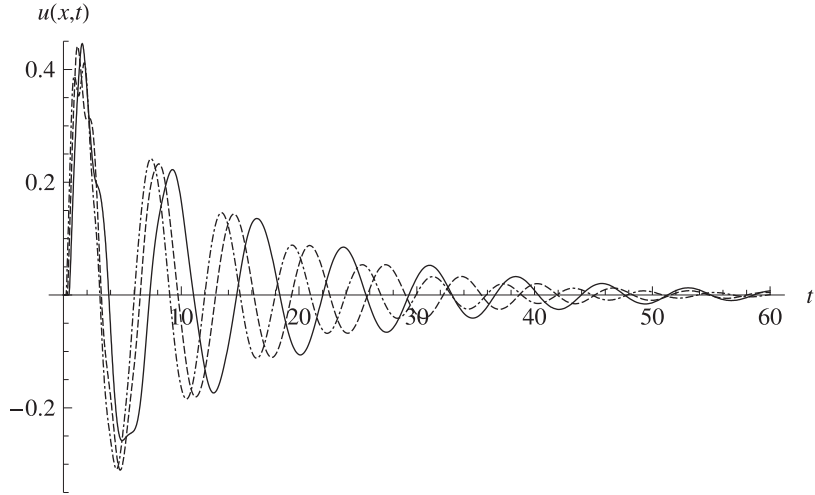}
\caption{Displacement $u(x,t)$ in the case $F=\protect\delta $ at $x=0.5$ as
a function of time $t\in (0,40)$ for $\varkappa =0.5$ - dot-dashed line, $%
\varkappa =1$ - dashed line and $\varkappa =2$ - solid line.}
\label{fig-7}
\end{figure}
One sees from Figure \ref{fig-7} that the larger the mass of a rod is (then
the value of $\varkappa $ is greater) the quasi-period (time between two
consecutive passage of a fixed point through its initial position) of the
oscillations is greater. That is due to the increased rod's inertia. Also,
for larger times there is no significant influence of $\varkappa $ on the
heights and widths of the peaks which indicates that the damping effects are
only due to the parameters $a,$ $b$ and $\alpha $ figuring in the
constitutive equation.
\begin{figure}[h]
\centering
\includegraphics[scale=1]{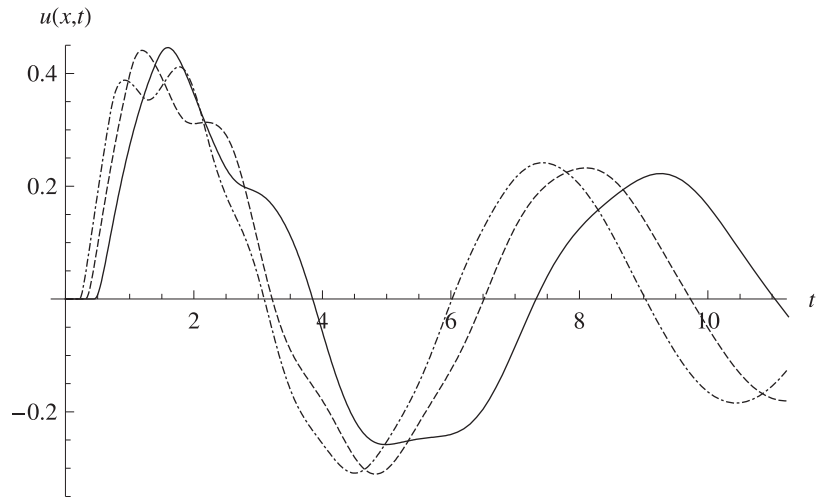}
\caption{Displacement $u(x,t)$ in the case $F=\protect\delta $ at $x=0.5$ as
a function of time $t\in (0,11)$ for $\varkappa =0.5$ - dot-dashed line, $%
\varkappa =1$ - dashed line and $\varkappa =2$ - solid line.}
\label{fig-8}
\end{figure}
Figure \ref{fig-8} shows that the shape of the curve in transitional regime
strongly depends on $\varkappa .$ Moreover, the delay in the oscillations
increases as the $\varkappa $ increases which indicates that the speed of
the wave propagation depends on the mass ratio.

\subsection{The case $F=H$}

The aim of this section is the qualitative analysis of the behavior of a
displacement $u$ when there is a force, given in the form of the Heaviside
function, i.e., $F=H,$ acting at the attached body. Thus, our results
correspond to a creep experiment. The rod is modelled by the fractional
Zener model, i.e., the function $M$ is given by (\ref{zener-M}). The
parameters describing the rod are: $a=0.1,$ $b=0.9,$ $\alpha =0.8.$ We
present plots of $u$ for the ratio between the masses of rod and body $%
\varkappa =1.$

Figure \ref{fig-9} present the long-time behavior of displacement $u$. One
notices that the rod creeps to a finite value of displacement so that $%
\lim_{t\rightarrow \infty }u\left( x,t\right) =x,$ $x\in \left[ 0,1\right] $%
. Figure \ref{fig-12} present the short-time behavior of $u$. We see that
there is a delay in starting time-instant of a point of a rod.
\begin{figure}[h]
\begin{minipage}{72mm}
 \centering
 \includegraphics[scale=0.7]{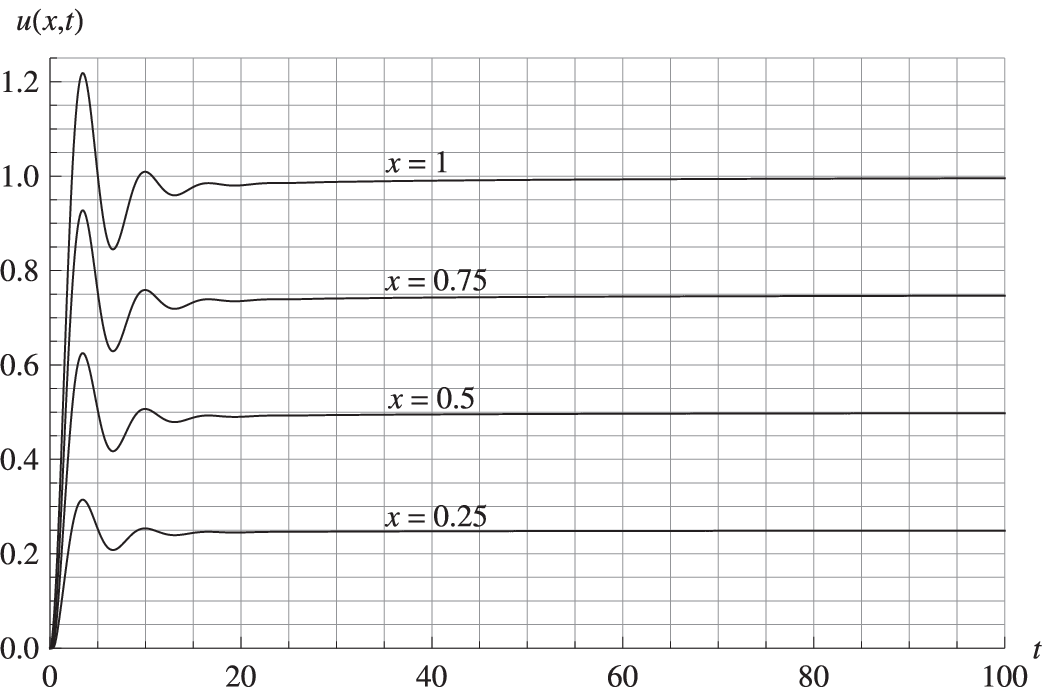}
 \caption{Displacement
 $u(x,t)$ in the creep experiment for $\varkappa=1$ as a function of time $t$ at
 $x\in \{0.25,0.5,0.75,1\}$ for $t\in (0,100)$.}
 \label{fig-9}
 \end{minipage}
\hfil
\begin{minipage}{72mm}
 \centering
 \includegraphics[scale=0.7]{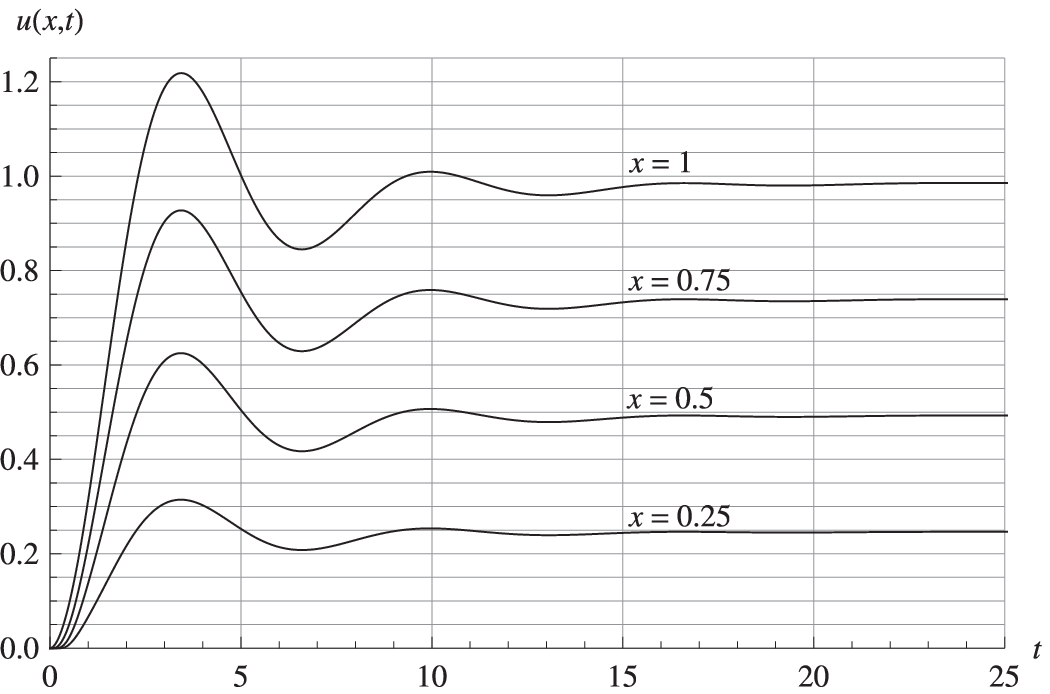}
 \caption{Displacement
 $u(x,t)$ in the creep experiment for $\varkappa=1$ as a function of time $t$ at
 $x\in \{0.25,0.5,0.75,1\}$ for $t\in (0,25)$.}
 \label{fig-12}
 \end{minipage}
\end{figure}

Figures \ref{fig-15} and \ref{fig-16} present the plots of time evolution of
displacement $u$ of a rod described by the Zener model for fixed point of
the rod $x=0.5$ if the ratio between masses $\varkappa $ varies. Here, the
parameters are $a=0.2,$ $b=0.6,$ $\alpha =0.45.$ One sees, Figure \ref%
{fig-15}, that there is no dependence of the finite value of displacement in
creep on the value of the mass ratio $\varkappa .$
\begin{figure}[h]
\centering
\includegraphics[scale=1]{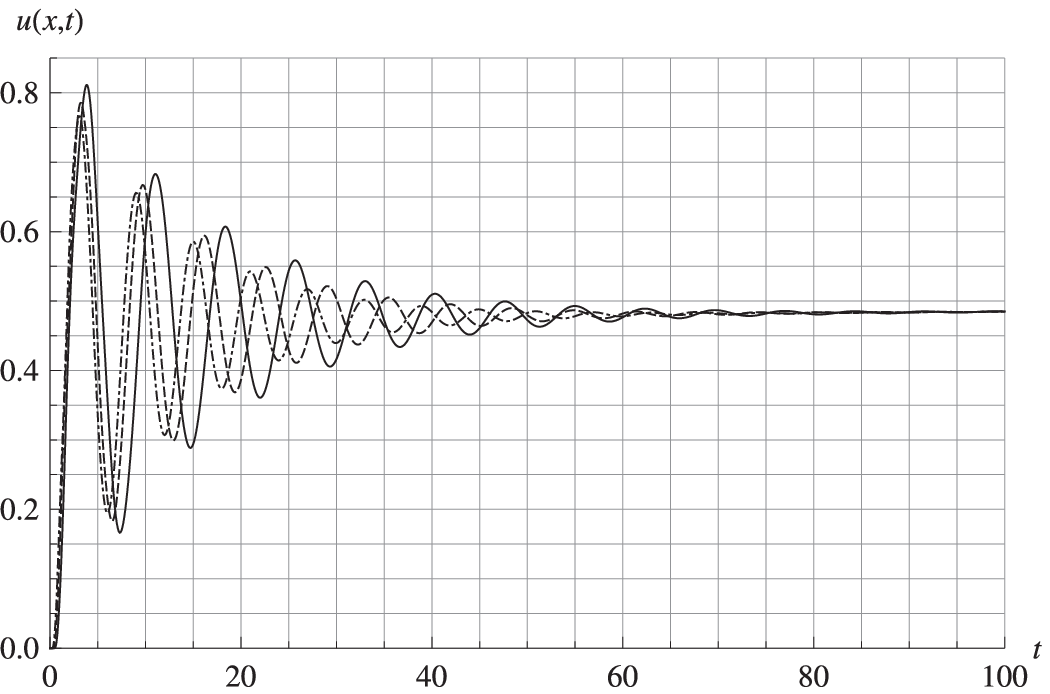}
\caption{Displacement $u(x,t)$ in the creep experiment at $x=0.5$ as a
function of time $t\in (0,100)$ for $\varkappa =0.5$ - dot-dashed line, $%
\varkappa =1$ - dashed line and $\varkappa =2$ - solid line.}
\label{fig-15}
\end{figure}
We see, Figure \ref{fig-16}, that for small times there is an influence of $%
\varkappa $ on the height of the peaks such that its height increases as $%
\varkappa $ increases. This is due to the inertia, while for larger times
the viscoelastic properties of the rod prevail. Similarly as in the previous
section the delay in the oscillations starting time-instant increases as $%
\varkappa $ increases.
\begin{figure}[h]
\centering
\includegraphics[scale=1]{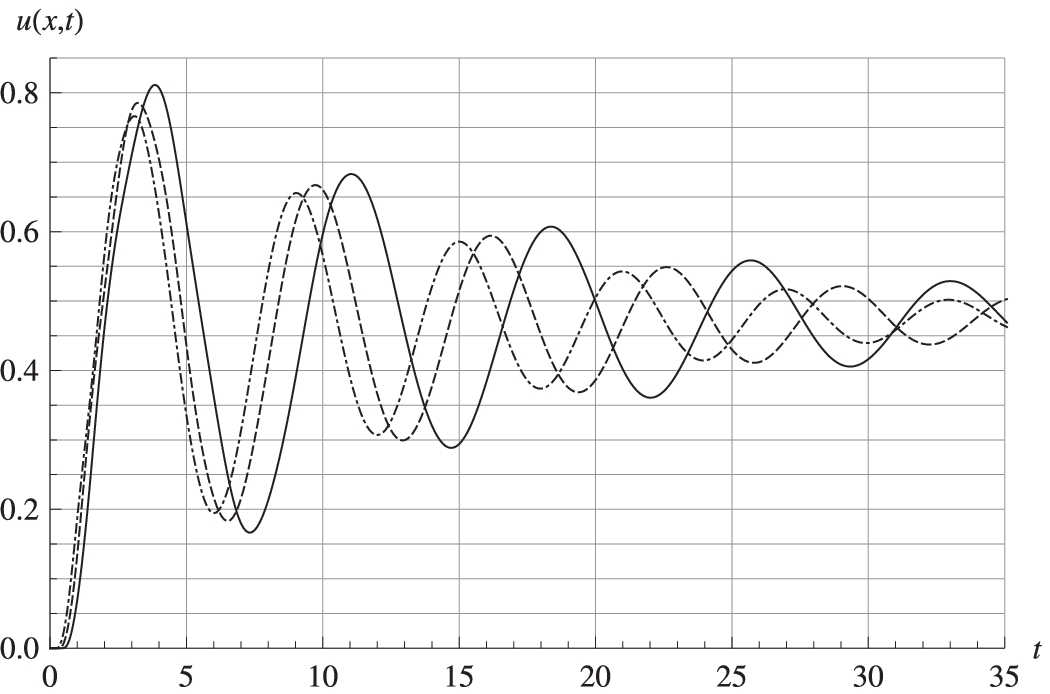}
\caption{Displacement $u(x,t)$ in the creep experiment at $x=0.5$ as a
function of time $t\in (0,35)$ for $\varkappa =0.5$ - dot-dashed line, $%
\varkappa =1$ - dashed line and $\varkappa =2$ - solid line.}
\label{fig-16}
\end{figure}

\begin{acknowledgement}
This research is supported by the Serbian Ministry of Education and Science
projects $174005$ (TMA and DZ) and $174024$ (SP), as well as by the
Secretariat for Science of Vojvodina project $114-451-2167$ (DZ).
\end{acknowledgement}

\end{document}